\documentclass[11pt]{article}

\usepackage[margin=1in]{geometry}
\usepackage[T1]{fontenc}
\usepackage[utf8]{inputenc}
\usepackage{lmodern}
\usepackage{microtype}
\usepackage{amsmath,amssymb,amsthm,mathtools}
\usepackage{booktabs,tabularx,array,float}
\usepackage{enumitem}
\usepackage[numbers,sort&compress]{natbib}
\usepackage[hidelinks]{hyperref}
\usepackage{xurl}
\hypersetup{
  pdftitle={Target-Dependent Local Verification: Information--Proof-Length Tradeoffs},
  pdfauthor={Hongmin Li}
}

\numberwithin{equation}{section}
\newtheorem{theorem}{Theorem}[section]
\newtheorem{lemma}[theorem]{Lemma}
\newtheorem{corollary}[theorem]{Corollary}
\newtheorem{proposition}[theorem]{Proposition}
\theoremstyle{definition}

\theoremstyle{remark}
\newtheorem{remark}[theorem]{Remark}
\newtheorem{example}[theorem]{Example}

\newcommand{\bits}{\{0,1\}}
\newcommand{\supp}{\operatorname{supp}}
\newcommand{\VCdim}{\operatorname{VCdim}}
\newcommand{\fibdim}{D_{\mathrm{fib}}}
\newcommand{\poly}{\operatorname{poly}}
\newcommand{\Prb}{\mathop{\Pr}}
\newcommand{\distH}{d_{\mathrm H}}
\newcommand{\Bbin}{\mathcal B}
\newcommand{\Lrldc}{\mathcal L}
\newcommand{\eps}{\varepsilon}

\title{Target-Dependent Local Verification:\\ Information--Proof-Length Tradeoffs}
\author{%
Hongmin Li\thanks{Corresponding author: \href{mailto:li.hongmin.xa@alumni.tsukuba.ac.jp}{li.hongmin.xa@alumni.tsukuba.ac.jp}; ORCID: \href{https://orcid.org/0000-0003-0228-0600}{0000-0003-0228-0600}}\\[0.4em]
\small School of Life Science and Technology, Institute of Science Tokyo\\
\small 2-12-1 Ookayama, Meguro-ku, Tokyo 152-8550, Japan\\[0.4em]
\small Department of Computational Biology and Medical Sciences,\\
\small Graduate School of Frontier Sciences, The University of Tokyo\\
\small 5-1-5 Kashiwanoha, Kashiwa-shi, Chiba 277-8561, Japan
}
\date{}

\begin{document}
\maketitle

\begin{abstract}
We study fixed-layout local verification with target-dependent local tests. Let $M$ be a random variable on $\bits^K$, and let $S$ record the test selected at each coordinate. For $s\in\supp(S)$, set $F_s=\{m\in\bits^K:\Prb[M=m,S=s]>0\}$, and let $D_{\mathrm{fib}}=\max_s\VCdim(F_s)$. Then
\[
  H(M\mid S)\le
  \log_2\!\left(\sum_{j=0}^{D_{\mathrm{fib}}}\binom Kj\right).
\]
If a fiber shatters $d$ coordinates, we can extract a weak relaxed locally decodable code with message length $d$ and block length $d+P$ over the original proof alphabet.

For a uniform $K$-bit target and fixed proof alphabet, $Q$, and $\sigma$, the nonadaptive lower bound of Goldberg--Gur--Saraogi implies that, for every fixed $0\le\gamma<1$, the condition $I(M;S)\le\gamma K$ forces
\[
 P=\Omega\!\left(
 \frac{K^{1+1/\lceil Q/\sigma\rceil}}
      {(\log K)^{2+2/\lceil Q/\sigma\rceil}}
 \right).
\]
If $P\le K(\log K)^c$ and $a=\lceil Q/\sigma\rceil$, the same argument gives
\[
 I(M;S)\ge K-
 O\!\left(K^{a/(a+1)}(\log K)^{3+ac/(a+1)}\right).
\]
In particular, $I(M;S)=K-o(K)$. Any discrete target-dependent verifier state $T$ that determines $S$ satisfies $I(M;T)\ge I(M;S)$.

Branches with bounded randomness and adaptive queries can be simulated nonadaptively by exposing the verifier's decision trees. If each branch uses at most $r$ random bits and $q$ adaptive proof queries, and the completeness--wrong-claim-soundness gap is fixed, exhaustive exposure gives a decoder with perfect completeness and at most $1+2^{r+1}\sum_{j<q}A^j$ queries. Under $I(M;S)\le\gamma K$ for a fixed $\gamma<1$, near-linear proof length requires this query complexity to be $\Omega(\log K/\log\log K)$. The binary one-query case gives $P=2^{\Omega(K)}$. Finally, we apply the explicit information bound to a global version of the list-sound dPCP interface of Gur--Minzer--Weissenberg--Zheng. A fixed target-independent menu of $L$ test profiles must satisfy $\log_2L\ge K-o(K)$. If the list depends on the proof, its description and any additional target-dependent selection data must be included in the state being measured.
\end{abstract}


\section{Introduction}

\subsection{Information carried by test selection}

A single proof word can support local verification of every coordinate of a target. The required information may reside in the proof itself or in the verifier's target-dependent choice of local tests. A lower bound stated only in terms of proof length accounts for the former but not the latter.

Let $M$ be a random target in $\bits^K$ and $\Pi(M)\in\Sigma^P$ its proof; all proof words share the same alphabet, length, and coordinate set $[P]$. At coordinate $i$, an index $\tau_i$ specifies the verifier branches for claims $0$ and $1$. We write
\[
  S=(\tau_1,\ldots,\tau_K)
\]
for the \emph{selected-test profile}. The mutual information $I(M;S)$ measures how much this profile reveals about the target. The selection may be randomized, with arbitrary correlations across coordinates. If target-dependent verifier state $T$ determines $S$, data processing gives $I(M;T)\ge I(M;S)$.

Conditioning on $S=s$ fixes all local tests and leaves a common-test fiber of targets. If $H(M\mid S)$ is large, one such fiber has large VC dimension. A shattered coordinate set then supplies a submessage. Choosing one target representative for each assignment, prepending the assignment to its proof, and using each prepended bit to select the corresponding claim branch produces a weak relaxed locally decodable code. Existing RLDC lower bounds therefore constrain either the proof length or the query complexity needed for a nonadaptive simulation.

Throughout, the proof alphabet and coordinate set are fixed, soundness holds pointwise in a positive Hamming neighborhood, and the extracted decoder has bounded query complexity.

\subsection{Main results}

Theorem~\ref{thm:entropy-fiber} bounds $H(M\mid S)$ by the Sauer--Shelah growth function at the largest VC dimension of a common-test fiber, while Lemma~\ref{lem:extraction} extracts a weak RLDC from any shattered fiber over the original proof alphabet.

Combining the extraction with the nonadaptive lower bound of Goldberg--Gur--Saraogi gives Theorem~\ref{thm:main}. When the proof length is near-linear, Theorem~\ref{thm:quantitative-info} gives an explicit sublinear upper bound on $H(M\mid S)$ and hence a lower bound on the information in any verifier state that determines $S$. Section~\ref{sec:finite} treats bounded randomness and adaptive proof queries by exposing the corresponding decision trees. The binary one-query case uses the exponential lower bound for weak adaptive two-query RLDCs.

Section~\ref{sec:dpcp} applies the same information bound to a global variant of the list-sound dPCP interface used by Gur--Minzer--Weissenberg--Zheng. Their interface permits different coordinates to be explained by different list members. If one instead requires the selected-test profile to come from a fixed target-independent menu, then a near-linear proof forces that menu to have size $2^{K-o(K)}$.

\subsection{Relation to prior work}

Classical locally decodable codes were systematically studied by Katz and Trevisan~\cite{KatzTrevisan2000}; Kerenidis and de Wolf proved the exponential lower bound for binary two-query codes~\cite{KerenidisDeWolf2004}. Relaxed locally decodable codes were introduced by Ben-Sasson, Goldreich, Harsha, Sudan, and Vadhan in their study of robust PCPs of proximity~\cite{BGHSV2006}. Gur and Lachish established general lower bounds for relaxed local decoding~\cite{GurLachish2021}, and Dall'Agnol, Gur, and Lachish extended the structural method to adaptive local algorithms~\cite{DallAgnolGurLachish2023}. We use the recent nonadaptive inequality of Goldberg, Gur, and Saraogi~\cite{GoldbergGurSaraogi2026} and the binary adaptive two-query lower bound of Block et al.~\cite{BlockEtAl2026}. Cheng, Li, and Mao obtain complementary results in a small-wrong-output regime by converting RLDCs to LDCs~\cite{ChengLiMao2026}.

The new ingredient is the reduction from fixed-layout target-dependent verification to relaxed local decoding. On a common-test fiber, the inline prefix converts target-relative wrong-claim soundness into relaxed decoding. The resulting bounds separate proof length from the information carried by local-test selection.

Several related models also account explicitly for auxiliary information. Applebaum and Nir's advisor--verifier--prover games charge function-dependent advice together with verifier communication~\cite{ApplebaumNir2023}. Boyle, Komargodski, and Vafa prove a local-space/query tradeoff for dynamic memory checking using an information-compression argument inspired by local-decoding lower bounds~\cite{BoyleKomargodskiVafa2025}. Our setting is static: the relevant auxiliary information is the target-dependent selection of local tests, and the reduction proceeds through the VC dimension of a conditional support.

Decodable PCPs were introduced as a composition interface by Dinur and Harsha~\cite{DinurHarsha2013}; structured two-query dPCPs also appear in the derandomized parallel-repetition framework of Dinur and Meir~\cite{DinurMeir2011}. Gur, Minzer, Weissenberg, and Zheng use a list-sound dPCP in their separation of three-query RLDCs from three-query LDCs~\cite{GurMinzerWeissenbergZheng2026}. Section~\ref{sec:dpcp} considers a stronger interface in which coordinatewise explanation is replaced by selection from a fixed target-independent menu of test profiles, together with target-relative robustness in a fixed proof layout.

To the best of our knowledge, existing RLDC lower bounds do not account for the mutual information carried by selected local tests, and the cited dPCP constructions do not state the target-relative fixed-layout interface considered here.

\subsection{Organization}

Section~\ref{sec:model} gives the model and the RLDC lower bound used later. Section~\ref{sec:entropy} relates conditional entropy to fiber VC dimension, and Section~\ref{sec:extraction} extracts an RLDC from a common-test fiber. Section~\ref{sec:quantitative} derives the explicit information bound for near-linear proofs. Section~\ref{sec:finite} treats finite randomness, Section~\ref{sec:dpcp} considers global selection in a list-sound dPCP interface, and Section~\ref{sec:scope} presents boundary examples and open problems.

\section{Model and coding input}\label{sec:model}

All logarithms and entropies are base two. Hamming distance is unnormalized unless a relative distance is stated explicitly.

\subsection{Fixed-layout target-dependent local verification}

Let $\Sigma$ be an alphabet of size $A\ge2$. Fix distinct symbols $a_0,a_1\in\Sigma$, and write $\iota(b)=a_b$ for $b\in\bits$. A \emph{fixed-layout proof map} is
\begin{equation}
  \Pi:\bits^K\longrightarrow\Sigma^P.
  \label{eq:proofmap}
\end{equation}
For each target $m$, the proof word $\Pi(m)$ may be different, but every proof word uses the same alphabet, length, and coordinate set $[P]$.

For each coordinate $i\in[K]$, claim $b\in\bits$, and local-test index $\tau$, let $V_{i,b}^{\tau}$ be a randomized verifier with oracle access to a word in $\Sigma^P$. A \emph{selected-test profile} is a discrete random variable
\begin{equation}
  S=(\tau_1,\ldots,\tau_K),
  \label{eq:test-profile}
\end{equation}
which may be randomized conditional on the target $M$. The index $\tau_i$ may encode both branch descriptions at coordinate $i$.

For every pair $(m,s)\in\supp(M,S)$, every $i\in[K]$, and fixed constants $0<\rho<1$ and $0\le\beta<1$, assume perfect completeness
\begin{equation}
  \Prb\!\left[V_{i,m_i}^{s_i}(\Pi(m))=1\right]=1,
  \label{eq:completeness}
\end{equation}
and target-relative wrong-claim soundness
\begin{equation}
  \distH(z,\Pi(m))\le \rho P
  \quad\Longrightarrow\quad
  \Prb\!\left[V_{i,1-m_i}^{s_i}(z)=1\right]\le\beta.
  \label{eq:soundness}
\end{equation}
Condition~\eqref{eq:soundness} is target-relative: it is imposed around the proof $\Pi(m)$ for the current target. Imposing the same wrong-claim condition around every target proof would conflict with perfect completeness whenever two targets differ at coordinate $i$. We require no robustness from the correct branch on a corrupted proof; rejection is interpreted as the relaxed-decoding symbol $\bot$.

Fixed layout means that every local test addresses the same coordinate set $[P]$; the proof words themselves may depend on the target. A target-dependent coordinate set or address translation is outside the model unless its description is included in the target-dependent verifier state.

\subsection{Common-test fibers and maximum fiber VC dimension}

For $s\in\supp(S)$, define the \emph{common-test fiber}
\begin{equation}
  F_s=\{m\in\bits^K:\Prb[M=m,S=s]>0\}.
  \label{eq:fiber}
\end{equation}
Its VC dimension is taken with respect to the $K$ target coordinates. Define the maximum fiber VC dimension
\begin{equation}
  \fibdim(M;S)=\max_{s\in\supp(S)}\VCdim(F_s).
  \label{eq:fiber-dim}
\end{equation}
For $-1\le d\le K$, write
\begin{equation}
  \Bbin_K(d)=\sum_{j=0}^{d}\binom Kj,
  \qquad \Bbin_K(-1)=0.
  \label{eq:binomial-growth}
\end{equation}
All targets in $F_s$ share the same selected-test profile. The quantity $\fibdim(M;S)$ is the largest number of target coordinates that can vary freely within a single fiber.

\subsection{Weak relaxed local decoding}

We use the weak relaxed-decoding convention. A code $C:\bits^k\to\Gamma^n$ has a nonadaptive $(Q,\delta,\sigma)$ relaxed decoder if, on input $j\in[k]$, the decoder makes at most $Q$ oracle queries chosen independently of the received word, satisfies
\[
  D^{C(x)}(j)=x_j
  \quad\text{with probability one},
\]
and, for every $w$ with $\distH(w,C(x))\le\delta n$,
\begin{equation}
  \Prb\!\left[D^w(j)\in\{x_j,\bot\}\right]\ge\sigma.
  \label{eq:weak-rldc}
\end{equation}
For $d\ge2$, define
\begin{equation}
  \Lrldc_{Q,A,\sigma}(d)=
  \left(
  \frac{\sigma^2d}
       {38Q^4\log_2^2(A)\log_2^2 d}
  \right)^{1+1/\lceil Q/\sigma\rceil}.
  \label{eq:Lfunction}
\end{equation}

We use the following form of the Goldberg--Gur--Saraogi lower bound~\cite[Theorem~1]{GoldbergGurSaraogi2026}. If a code with $k$ message bits, block length $n$, alphabet $\Gamma$, and a nonadaptive $(Q,\delta,\sigma)$ weak relaxed decoder satisfies
\begin{equation}
  \delta>n^{-\sigma/(2Q)},
  \label{eq:GGSradius}
\end{equation}
then
\begin{equation}
  \frac{k}{\log_2^2 k}
  \le
  38Q^4\sigma^{-2}\log_2^2|\Gamma|\,
  n^{1-1/(\lceil Q/\sigma\rceil+1)}.
  \label{eq:GGS}
\end{equation}
Equivalently, $n\ge\Lrldc_{Q,|\Gamma|,\sigma}(k)$.

\section{Conditional entropy and common-test fibers}\label{sec:entropy}

The entropy bound does not require $M$ to be uniform.

\begin{theorem}[Conditional entropy versus fiber VC dimension]\label{thm:entropy-fiber}
For every random variable $M$ supported on $\bits^K$ and every selected-test profile $S$,
\begin{equation}
  H(M\mid S)
  \le
  \log_2 \Bbin_K\!\bigl(\fibdim(M;S)\bigr).
  \label{eq:entropy-fiber}
\end{equation}
Equivalently, if an integer $d\in[K]$ satisfies
\begin{equation}
  H(M\mid S)>\log_2\Bbin_K(d-1),
  \label{eq:entropy-threshold}
\end{equation}
then some common-test fiber shatters at least $d$ target coordinates.
\end{theorem}

\begin{proof}
Let $D=\fibdim(M;S)$. For every $s\in\supp(S)$,
\[
  H(M\mid S=s)\le\log_2|F_s|.
\]
Since $\VCdim(F_s)\le D$, the Sauer--Shelah bound~\cite{Sauer1972} gives
\[
  |F_s|\le\sum_{j=0}^{D}\binom Kj=\Bbin_K(D).
\]
Averaging over $S$ proves~\eqref{eq:entropy-fiber}. If every fiber had VC dimension at most $d-1$, the same argument would give $H(M\mid S)\le\log_2\Bbin_K(d-1)$, contradicting~\eqref{eq:entropy-threshold}.
\end{proof}

For a uniform target, Theorem~\ref{thm:entropy-fiber} yields a linear-dimensional fiber whenever a constant fraction of the target entropy remains after conditioning on $S$. Let $h_2(t)=-t\log_2t-(1-t)\log_2(1-t)$ be the binary entropy function.

\begin{corollary}[A linear-dimensional common-test fiber]\label{cor:linear-fiber}
Let $M$ be uniform on $\bits^K$ and suppose
\begin{equation}
  I(M;S)\le\gamma K
  \label{eq:info-gamma}
\end{equation}
for a fixed $0\le\gamma<1$. Let $\vartheta_\gamma\in(0,1/2]$ satisfy
\[
  h_2(\vartheta_\gamma)=1-\gamma.
\]
For every fixed $0<\eps<\vartheta_\gamma$ and all sufficiently large $K$, some common-test fiber shatters
\begin{equation}
  d=\left\lfloor(\vartheta_\gamma-\eps)K\right\rfloor
  \label{eq:dlinear}
\end{equation}
target coordinates.
\end{corollary}

\begin{proof}
Uniformity and~\eqref{eq:info-gamma} give
\[
 H(M\mid S)=K-I(M;S)\ge(1-\gamma)K.
\]
For $d$ as in~\eqref{eq:dlinear},
\[
 \log_2\Bbin_K(d-1)
 \le \bigl(h_2(\vartheta_\gamma-\eps)+o(1)\bigr)K
 <(1-\gamma)K.
\]
Apply Theorem~\ref{thm:entropy-fiber}.
\end{proof}

If $M$ has full support and $S$ is independent of $M$, every nonempty fiber is the full cube and $\fibdim(M;S)=K$. At the opposite extreme, $S=M$ gives $\fibdim(M;S)=0$.

\section{Extracting an RLDC from a common-test fiber}\label{sec:extraction}

Fix $s\in\supp(S)$ and a set $I=\{i_1,\ldots,i_d\}$ shattered by $F_s$. For every $x\in\bits^d$, choose a representative $m_x\in F_s$ satisfying $m_x|_I=x$, and define
\begin{equation}
  C_s(x)=\Pi(m_x)\in\Sigma^P.
  \label{eq:suffix-code}
\end{equation}
Set
\begin{equation}
  d_0(A)=\max\!\left\{2,\left\lceil2\log_2A\right\rceil\right\}.
  \label{eq:d0}
\end{equation}

\begin{lemma}[RLDC extraction from a common-test fiber]\label{lem:extraction}
Assume $d\ge d_0(A)$. Under~\eqref{eq:completeness}--\eqref{eq:soundness}, the map $C_s$ is injective and its distinct codewords have Hamming distance greater than $\rho P$. Moreover, the inline code
\begin{equation}
  \widehat C_s(x)=
  \iota(x_1)\cdots\iota(x_d)\Vert C_s(x)
  \in\Sigma^{d+P}
  \label{eq:inline-code}
\end{equation}
has a weak relaxed decoder with perfect completeness.

The decoder uses one inline query, followed by the queries of the claim branch selected by the inline symbol. If both claim branches are nonadaptive and make at most $q$ proof queries, the extracted decoder is nonadaptive with
\begin{equation}
  Q=1+2q
  \label{eq:extracted-Q}
\end{equation}
queries and success parameter
\begin{equation}
  \sigma=1-\beta.
  \label{eq:extracted-sigma}
\end{equation}
It has decoding radius
\begin{equation}
  \delta_0=
  \frac{\rho}{4\bigl(1+(1-\rho)\log_2 A\bigr)}.
  \label{eq:delta0}
\end{equation}
\end{lemma}

\begin{proof}
Suppose $x_j\ne x'_j$ and $\distH(C_s(x),C_s(x'))\le\rho P$. The same local-test index $s_{i_j}$ is used throughout the fiber. By perfect completeness for $m_{x'}$, the branch claiming $x'_j$ accepts $C_s(x')$ with probability one. Yet $C_s(x')$ lies within distance $\rho P$ of $C_s(x)=\Pi(m_x)$, so target-relative soundness for $m_x$ bounds the same acceptance probability by $\beta<1$. This contradiction shows that distinct proof suffixes have distance greater than $\rho P$; in particular, $C_s$ is injective.

To decode $x_j$, read inline coordinate $j$. If its value is $a_b$ for $b\in\bits$, run $V_{i_j,b}^{s_{i_j}}$ on the proof suffix and output $b$ on acceptance and $\bot$ on rejection. If the inline value is in $\Sigma\setminus\{a_0,a_1\}$, output $\bot$ without invoking a branch. On an uncorrupted codeword, the inline claim is correct and~\eqref{eq:completeness} gives perfect completeness.

Now suppose the proof suffix of a received word $w$ lies within distance $\rho P$ of $C_s(x)$. If inline coordinate $j$ is unchanged, the decoder outputs $x_j$ or $\bot$. If it is changed to $a_{1-x_j}$, the wrong-claim branch is selected and the wrong bit is output with probability at most $\beta$ by~\eqref{eq:soundness}. Any other inline symbol produces $\bot$. Under this suffix-distance condition, the decoder therefore satisfies~\eqref{eq:weak-rldc} with $\sigma=1-\beta$.

If both claim branches are nonadaptive, sample one random tape for each claim and query the inline coordinate together with the union of the two proof-query sets. After reading the inline symbol, retain only the answers for the selected branch. All query locations are fixed in advance, and there are at most $1+2q$ of them.

It remains to control the relative radius. Let $D_{\min}$ be the minimum distance among the $2^d$ proof suffixes. The $A$-ary Singleton bound and $D_{\min}>\rho P$ give
\begin{equation}
  d\le(P-D_{\min}+1)\log_2 A
  <\bigl((1-\rho)P+1\bigr)\log_2 A.
  \label{eq:singleton}
\end{equation}
If $d\ge d_0(A)$, then $d\ge2\log_2A$ and hence
\[
  P>\frac{d-\log_2A}{(1-\rho)\log_2A}
  \ge \frac{d}{2(1-\rho)\log_2A}.
\]
Consequently, the value in~\eqref{eq:delta0} satisfies $\delta_0(d+P)<\rho P/2$. Every word within relative distance $\delta_0$ therefore satisfies the suffix-distance condition used above. Since $D_{\min}>\rho P$, the same inequality places the received word within less than half the minimum code distance, so the nearby codeword is unique.
\end{proof}

\begin{remark}[Existential nature of the extraction]
The extraction is existential: neither the shattered set nor the representatives $m_x$ need be found efficiently. The lemma controls only block length and query complexity.
\end{remark}

Because $a_0,a_1\in\Sigma$, the received word remains in $\Sigma^{d+P}$ and its proof suffix remains in the verifier's oracle domain $\Sigma^P$. No alphabet extension is needed.

\begin{theorem}[General fiber-to-RLDC lower bound]\label{thm:fiber-rldc}
Assume~\eqref{eq:completeness}--\eqref{eq:soundness}, and suppose both claim branches indexed by the selected-test profile are nonadaptive with at most $q$ proof queries. Set $Q=1+2q$ and $\sigma=1-\beta$. Let $d\ge d_0(A)$ and assume
\[
  d\le\fibdim(M;S).
\]
A sufficient information-theoretic condition is
\[
  H(M\mid S)>\log_2\Bbin_K(d-1).
\]
If
\begin{equation}
  \delta_0>(d+P)^{-\sigma/(2Q)},
  \label{eq:fiber-radius}
\end{equation}
then
\begin{equation}
  d+P\ge\Lrldc_{Q,A,\sigma}(d).
  \label{eq:fiber-bound}
\end{equation}
\end{theorem}

\begin{proof}
The condition $d\le\fibdim(M;S)$ directly supplies a $d$-coordinate shattered subset of some fiber; the displayed entropy condition supplies one by Theorem~\ref{thm:entropy-fiber}. Lemma~\ref{lem:extraction} then gives a code of message length $d$ and block length $d+P$ over $\Sigma$, with nonadaptive query complexity $Q$, success parameter $\sigma$, and radius $\delta_0$. Conditions~\eqref{eq:fiber-radius} and~\eqref{eq:GGSradius} coincide, so~\eqref{eq:GGS} gives~\eqref{eq:fiber-bound}.
\end{proof}

\begin{theorem}[Information--proof-length tradeoff]\label{thm:main}
Let $M$ be uniform on $\bits^K$ and suppose~\eqref{eq:completeness}--\eqref{eq:soundness} hold. Assume both claim branches indexed by the selected-test profile are nonadaptive with at most $q$ proof queries, and set $Q=1+2q$ and $\sigma=1-\beta$. If $I(M;S)\le\gamma K$ for a fixed $0\le\gamma<1$, then for every fixed $0<\eps<\vartheta_\gamma$ and $d$ given by~\eqref{eq:dlinear},
\begin{equation}
  d+P\ge\Lrldc_{Q,A,\sigma}(d)
  \label{eq:main-finite}
\end{equation}
whenever $d\ge d_0(A)$ and~\eqref{eq:fiber-radius} holds. Consequently, for fixed $A,Q,\gamma,\eps$, and $\sigma$,
\begin{equation}
  P=\Omega\!\left(
  \frac{K^{1+1/\lceil Q/\sigma\rceil}}
       {(\log K)^{2+2/\lceil Q/\sigma\rceil}}
  \right).
  \label{eq:main-asymptotic}
\end{equation}
\end{theorem}

\begin{proof}
Corollary~\ref{cor:linear-fiber} supplies $d=\Theta(K)$ shattered coordinates, so $d\ge d_0(A)$ for all sufficiently large $K$. Apply Theorem~\ref{thm:fiber-rldc}. For fixed $Q$ and $\sigma$, the constant $\delta_0$ satisfies~\eqref{eq:fiber-radius} for all sufficiently large $K$. Writing $a=\lceil Q/\sigma\rceil$, the right-hand side of~\eqref{eq:main-finite} equals
\[
  \Omega\!\left(
  \frac{d^{1+1/a}}{(\log d)^{2+2/a}}
  \right),
\]
which is superlinear in $d$. Subtracting $d$ changes only the implicit constant, and $d=\Theta(K)$ gives~\eqref{eq:main-asymptotic}.
\end{proof}

\begin{corollary}[Binary one-query endpoint]\label{cor:binary}
Suppose $A=2$, $\beta<1/2$, every branch indexed by the selected-test profile makes one possibly adaptive proof query, and $I(M;S)\le\gamma K$ for a uniform target and a fixed $0\le\gamma<1$. Then
\begin{equation}
  P=2^{\Omega_{\rho,\beta,\gamma}(K)}.
  \label{eq:binary-exp}
\end{equation}
\end{corollary}

\begin{proof}
On the shattered fiber supplied by Corollary~\ref{cor:linear-fiber}, the inline code is binary because $\Sigma=\{a_0,a_1\}$. Its decoder is weak and adaptive, has perfect completeness, and makes two queries, one to the inline prefix and one to the proof suffix. The success probability is $1-\beta>1/2$. The radius in~\eqref{eq:delta0} is below half the minimum relative distance, so the nearby codeword is unique. The weak adaptive two-query lower bound of Block et al.~\cite[Theorem~1]{BlockEtAl2026} gives exponential block length. Since the inline prefix has length $d=O(K)$, an exponential lower bound on $d+P$ implies $P=2^{\Omega(K)}$.
\end{proof}
\section{Test-selection information required by near-linear proofs}\label{sec:quantitative}

Theorem~\ref{thm:main} treats the case $H(M\mid S)=\Omega(K)$. For near-linear proof length, applying Theorem~\ref{thm:fiber-rldc} at $d=\fibdim(M;S)$ yields an explicit bound even when the residual entropy is sublinear.

\begin{theorem}[Residual entropy under near-linear proof length]\label{thm:quantitative-info}
Fix $A,Q,\rho,\sigma$, with $A\ge2$, $Q\ge1$, $0<\rho<1$, and $0<\sigma\le1$, and put
\begin{equation}
  a=\left\lceil\frac{Q}{\sigma}\right\rceil.
  \label{eq:a-def}
\end{equation}
Consider a family of fixed-layout verifiers satisfying~\eqref{eq:completeness}--\eqref{eq:soundness} such that, on every common-test fiber, the extraction in Lemma~\ref{lem:extraction} is nonadaptive with at most $Q$ queries and success parameter at least $\sigma$. If, for a fixed $c\ge0$,
\begin{equation}
  P\le K(\log_2 K)^c,
  \label{eq:nearlinear}
\end{equation}
then, for all sufficiently large $K$,
\begin{equation}
  H(M\mid S)
  \le
  C\,K^{a/(a+1)}
  (\log_2 K)^{3+ac/(a+1)},
  \label{eq:residual-explicit}
\end{equation}
where $C$ depends only on $A,Q,\rho,\sigma$, and $c$.

If $M$ is uniform on $\bits^K$, then
\begin{equation}
  I(M;S)
  \ge
  K-C\,K^{a/(a+1)}
  (\log_2 K)^{3+ac/(a+1)}.
  \label{eq:info-explicit}
\end{equation}
\end{theorem}

\begin{proof}
Let $D=\fibdim(M;S)$. If $D<d_0(A)$, Theorem~\ref{thm:entropy-fiber} gives $H(M\mid S)=O_{A}(\log K)$, which is stronger than~\eqref{eq:residual-explicit}. We may therefore assume $D\ge d_0(A)$.

If the radius condition fails at $d=D$, then
\[
  \delta_0\le(D+P)^{-\sigma/(2Q)},
\]
so
\[
  D+P\le\delta_0^{-2Q/\sigma}=O_{A,Q,\rho,\sigma}(1).
\]
Thus $D=O(1)$, and Theorem~\ref{thm:entropy-fiber} gives $H(M\mid S)=O(\log K)$.

Otherwise Theorem~\ref{thm:fiber-rldc}, applied with $d=D$, yields
\begin{equation}
  D+P
  \ge
  c_0\frac{D^{1+1/a}}{(\log_2 D)^{2+2/a}},
  \label{eq:D-lower}
\end{equation}
for a constant $c_0=c_0(A,Q,\sigma)>0$. Since $D\le K$ and~\eqref{eq:nearlinear} holds,
\[
  D+P\le2K(\log_2 K)^c
\]
for all sufficiently large $K$. Combining this with~\eqref{eq:D-lower} gives
\begin{equation}
  D
  \le
  C_1 K^{a/(a+1)}
  (\log_2 K)^{2+ac/(a+1)}.
  \label{eq:D-upper}
\end{equation}
The Sauer--Shelah estimate
\[
  \Bbin_K(D)\le\left(\frac{eK}{D}\right)^D
  \qquad (1\le D\le K)
\]
and Theorem~\ref{thm:entropy-fiber} imply
\[
  H(M\mid S)
  \le D\log_2\!\left(\frac{eK}{D}\right)
  \le D\log_2(eK).
\]
Substituting~\eqref{eq:D-upper} proves~\eqref{eq:residual-explicit}. For uniform $M$, use $I(M;S)=K-H(M\mid S)$.
\end{proof}

For fixed $Q$ and $\sigma$, the error term in~\eqref{eq:info-explicit} is $o(K)$. Thus $I(M;S)=K-o(K)$, and the error term records the dependence on the extracted decoder's query complexity and on the proof overhead.

\begin{corollary}[Selector support and verifier-state size]\label{cor:configuration}
Under the hypotheses of Theorem~\ref{thm:quantitative-info} with uniform $M$, define
\[
  R_{a,c}(K)=K^{a/(a+1)}(\log_2 K)^{3+ac/(a+1)}.
\]
Then
\begin{equation}
  |\supp(S)|\ge 2^{K-O(R_{a,c}(K))}.
  \label{eq:support-lower}
\end{equation}
More generally, suppose target-dependent verifier state $T$ determines the selected-test profile, so that $S=f(T)$. Then
\begin{equation}
  I(M;T)\ge I(M;S)\ge K-O(R_{a,c}(K)).
  \label{eq:config-info}
\end{equation}
If $T$ has at most $2^b$ possible values, then
\begin{equation}
  b\ge K-O(R_{a,c}(K)).
  \label{eq:config-bits}
\end{equation}
\end{corollary}

\begin{proof}
Theorem~\ref{thm:quantitative-info} and
\[
  I(M;S)\le H(S)\le\log_2|\supp(S)|
\]
give~\eqref{eq:support-lower}. If $S=f(T)$, the Markov chain $M\to T\to S$ and data processing give $I(M;T)\ge I(M;S)$. Finally, $I(M;T)\le H(T)\le b$ when $|\supp(T)|\le2^b$.
\end{proof}

\begin{proposition}[Information in target-dependent verifier state]\label{prop:capture}
Suppose an implementation has a fixed-layout proof map $\Pi:\bits^K\to\Sigma^P$ and discrete target-dependent verifier state $T$. For each coordinate $i$, let the local-test index used by both claim branches be $g_i(T)$, where the index may encode both branch descriptions. Assume that, for every $(m,t)\in\supp(M,T)$ and every coordinate $i$, the induced branches satisfy the completeness and target-relative soundness conditions~\eqref{eq:completeness}--\eqref{eq:soundness}. Then
\[
  S=(g_1(T),\ldots,g_K(T))
\]
defines a selected-test profile satisfying $I(M;S)\le I(M;T)$. Hence any information lower bound for $S$ also applies to $T$, whether or not the implementation stores the profile explicitly.
\end{proposition}

\begin{proof}
The selected-test profile is a deterministic function of $T$, so the variables form the Markov chain $M\to T\to S$. Data processing gives $I(M;S)\le I(M;T)$.
\end{proof}

The state $T$ may be an advice string, a seed, local state, a routing table, or a circuit description, provided it determines the selected tests. The conclusion concerns only the target-dependent state represented by $T$.

\section{Finite randomness and nonadaptive exposure}\label{sec:finite}

We now allow imperfect completeness. Fix constants $0\le\beta<\alpha\le1$ and
\begin{equation}
  \Prb\!\left[V_{i,m_i}^{s_i}(\Pi(m))=1\right]
  \ge\alpha>\beta,
  \label{eq:imperfect-completeness}
\end{equation}
and suppose each branch uses at most $r$ random bits and at most $q\ge1$ adaptive proof queries. Define
\begin{equation}
  T_q(A)=\sum_{j=0}^{q-1}A^j,
  \qquad
  Q_{\mathrm{exp}}=1+2^{r+1}T_q(A).
  \label{eq:Qexp}
\end{equation}

\begin{lemma}[Nonadaptive simulation by decision-tree exposure]\label{lem:decision-tree-exposure}
On every shattered common-test fiber of dimension $d\ge d_0(A)$, the inline code~\eqref{eq:inline-code} has a nonadaptive relaxed decoder with perfect completeness, zero wrong-output probability in the promise ball, and query complexity at most $Q_{\mathrm{exp}}$.
\end{lemma}

\begin{proof}
The distance argument in Lemma~\ref{lem:extraction} remains valid: if two suffixes encoding different requested bits were within distance $\rho P$, the same claim branch would have acceptance probability at least $\alpha$ by~\eqref{eq:imperfect-completeness} and at most $\beta$ by~\eqref{eq:soundness}. Hence the constant inline radius $\delta_0$ remains available.

For each claim $b\in\bits$ and each of at most $2^r$ random tapes, expose every internal node of the corresponding depth-$q$, $A$-ary proof-answer decision tree. Each tree contains at most $T_q(A)$ queried nodes. Query the inline coordinate and the union of all exposed proof positions for both claims. This uses at most $1+2^{r+1}T_q(A)$ queries, all chosen independently of the received answers.

Once the answers are known, an inline symbol $a_b$ selects the corresponding claim. Simulating that claim on all random tapes then gives its exact acceptance probability. If the inline symbol is outside $\{a_0,a_1\}$, output $\bot$. Otherwise use the threshold
\[
  t=\frac{\alpha+\beta}{2}.
\]
Output $b$ if the exact acceptance probability exceeds $t$, and output $\bot$ otherwise. On an honest word, the correct branch has acceptance probability greater than $t$, so the new decoder has perfect completeness. In the promise ball, changing the inline symbol to the opposite bit symbol selects a wrong-claim branch with acceptance probability at most $\beta<t$ and therefore yields $\bot$. An unchanged inline symbol yields the correct bit or $\bot$.
\end{proof}

This is an information-theoretic simulation: running time is unrestricted, and only the number of nonadaptive oracle queries is counted.

\begin{theorem}[Near-linear proofs require large exposure complexity]\label{thm:exposure}
Let $M$ be uniform on $\bits^K$, and assume~\eqref{eq:soundness} and~\eqref{eq:imperfect-completeness}. Fix $A\ge2$, $0<\rho<1$, $0\le\beta<\alpha\le1$, $0\le\gamma<1$, $0<\eps<\vartheta_\gamma$, and $c\ge0$. Suppose $I(M;S)\le\gamma K$, while $r=r(K)$ and $q=q(K)\ge1$ may vary with $K$. If
\begin{equation}
  P\le K(\log_2 K)^c
  \label{eq:nearlinear-finite}
\end{equation}
along an infinite sequence, then
\begin{equation}
  Q_{\mathrm{exp}}
  =\Omega\!\left(\frac{\log K}{\log\log K}\right),
  \label{eq:Qexp-lower}
\end{equation}
and therefore
\begin{equation}
  r+\log_2T_q(A)
  \ge
  \log_2\log K-\log_2\log\log K-O(1).
  \label{eq:r-q-lower}
\end{equation}
\end{theorem}

\begin{proof}
Corollary~\ref{cor:linear-fiber} supplies $d=\Theta(K)$ shattered coordinates, so $d\ge d_0(A)$ for all sufficiently large $K$. Put $N=d+P$. By~\eqref{eq:nearlinear-finite},
\begin{equation}
  N\le C_1d(\log_2d)^c.
  \label{eq:N-upper}
\end{equation}
Lemma~\ref{lem:decision-tree-exposure} gives a nonadaptive decoder with perfect completeness, $Q_{\mathrm{exp}}$ queries, and success parameter $1$.

If the GGS radius condition holds, then
\begin{equation}
  N\ge
  \left(
  \frac{d}
       {38Q_{\mathrm{exp}}^4\log_2^2A\log_2^2d}
  \right)^{1+1/Q_{\mathrm{exp}}}.
  \label{eq:N-lower-finite}
\end{equation}
After division by $d$ and taking base-two logarithms,~\eqref{eq:N-lower-finite} gives
\[
  \log_2(N/d)
  \ge
  \frac{\log_2 d}{Q_{\mathrm{exp}}}
  -O\!\left(\log Q_{\mathrm{exp}}+\log\log d\right).
\]
On the other hand,~\eqref{eq:N-upper} gives $\log_2(N/d)\le O(\log\log d)$. If $Q_{\mathrm{exp}}\le\eta\log d/\log\log d$ for a sufficiently small constant $\eta=\eta(A,c)>0$, the first term in the lower bound dominates the $O(\log\log d)$ terms, a contradiction.

If the radius condition fails, then
\[
  \delta_0\le N^{-1/(2Q_{\mathrm{exp}})},
\]
which implies
\[
  Q_{\mathrm{exp}}
  \ge\frac{\log N}{2\log(1/\delta_0)}
  =\Omega_{A,\rho}(\log K).
\]
Both cases prove~\eqref{eq:Qexp-lower}. Taking base-two logarithms in~\eqref{eq:Qexp} gives~\eqref{eq:r-q-lower}.
\end{proof}

The quantity $Q_{\mathrm{exp}}$ counts queries made by the exhaustive nonadaptive simulation, not by the original adaptive verifier. A different adaptive decoder may use fewer queries.

\section{Global selection in list-sound decodable PCPs}\label{sec:dpcp}

For every $\epsilon>0$, Gur, Minzer, Weissenberg, and Zheng construct a two-query decodable PCP with near-linear length, constant alphabet, perfect completeness, and $(L,\epsilon)$ list-decoding soundness, where $L=\poly(1/\epsilon)$~\cite[Theorem~3.12]{GurMinzerWeissenbergZheng2026}. In the notation of their Definition~3.3, the relevant quantifiers are
\begin{equation}
  \forall T_A\;\exists W(T_A),\ |W(T_A)|\le L\;\forall T_B.
  \label{eq:GMWZ-quantifiers}
\end{equation}
If $W(T_A)=\{w_1,\ldots,w_L\}$, their soundness guarantee bounds the probability of simultaneous acceptance and a decoded symbol outside
\[
  \{(w_1)_t,\ldots,(w_L)_t\}
\]
at the sampled coordinate $t$. Thus the decoded symbol need only lie in the coordinate projection of the list. Different coordinates may be explained by different list members; the guarantee does not choose one target-bound word that simultaneously explains every coordinate.

Consider instead a global rule that chooses the selected-test profile from a fixed menu independent of both the target and the proof word. The information bound then controls the size of that menu.

\begin{corollary}[A fixed global menu requires target-scale information]\label{cor:fixed-menu}
Fix $A\ge2$, $0<\rho<1$, $Q\ge1$, $0<\sigma\le1$, and $c\ge0$. Consider a strengthening of the list-sound interface for a uniform $K$-bit target $M$ with the following properties:
\begin{enumerate}[leftmargin=2.1em,label=\textup{(\alph*)}]
\item each target has a proof word $\Pi(M)\in\Sigma^P$ in a fixed proof layout, where $|\Sigma|=A$;
\item the target-relative conditions~\eqref{eq:completeness}--\eqref{eq:soundness} hold with $\beta=1-\sigma$;
\item both claim branches are nonadaptive and the extracted decoder has query complexity at most $Q$; and
\item the selected-test profile $S$ takes values in a fixed set $\mathcal S=\{s^{(1)},\ldots,s^{(L)}\}$ that is independent of the target and proof word.
\end{enumerate}
If $P\le K(\log_2K)^c$ and $a=\lceil Q/\sigma\rceil$, then
\begin{equation}
  \log_2 L
  \ge
  K-O\!\left(
  K^{a/(a+1)}(\log_2K)^{3+ac/(a+1)}
  \right).
  \label{eq:list-lower}
\end{equation}
In particular, $L\ge2^{K-o(K)}$.
\end{corollary}

\begin{proof}
Since $|\supp(S)|\le L$,
\[
  I(M;S)\le H(S)\le\log_2L.
\]
The corresponding lower bound on $I(M;S)$ follows from Theorem~\ref{thm:quantitative-info}, proving~\eqref{eq:list-lower}.
\end{proof}

In the GMWZ construction, however, the list $W(T_A)$ may depend on $T_A$, and the coordinatewise list-decoding quantifier remains in place. An index $J\in[L]$ therefore need not specify the same selected-test profile for different targets. If a global mechanism has $S=f(J,W(T_A),T)$ for additional target-dependent state $T$, the information bound applies to the joint state $(J,W(T_A),T)$, not to $J$ alone.

\section{Boundary examples, scope, and open directions}\label{sec:scope}

\subsection{The two information endpoints}

At one endpoint, the test selection is independent of the target. If $S$ is independent of $M$, then every nonempty common-test fiber equals $\supp(M)$. For a full-support target distribution, each such fiber is the entire cube and the RLDC extraction may use all $K$ target coordinates.

At the other endpoint, the selected tests can carry the entire target and make the proof essentially vacuous.

\begin{example}[Full-information selected-test profile and constant proof]\label{ex:full-info}
Fix any $a_\star\in\Sigma$, let $\Pi(m)=a_\star\in\Sigma^1$, and set the selected local-test index at coordinate $i$ to $s_i=m_i$. Define $V_{i,b}^{s_i}$ to ignore its oracle and accept exactly when $b=s_i$. Then~\eqref{eq:completeness}--\eqref{eq:soundness} hold with $P=1$, $q=0$, and every $0<\rho<1$, $0\le\beta<1$, while
\begin{equation}
  S=M,
  \qquad I(M;S)=K,
  \qquad \fibdim(M;S)=0.
  \label{eq:full-info-endpoint}
\end{equation}
\end{example}

The lower bound and Example~\ref{ex:full-info} leave a sublinear window below $K$ bits of test-selection information. The explicit deficit in Theorem~\ref{thm:quantitative-info} need not be tight, and the behavior inside that window remains open.

\subsection{Parameter regimes and limitations}

\begin{table}[H]
\centering
\footnotesize
\begin{tabularx}{\textwidth}{@{}>{\raggedright\arraybackslash}p{0.24\textwidth}>{\raggedright\arraybackslash}p{0.28\textwidth}>{\raggedright\arraybackslash}X@{}}
\toprule
Selected verifier branches & Extracted decoder & Consequence under residual target entropy \\
\midrule
Nonadaptive, at most $q$ proof queries per claim
& Nonadaptive, $Q=1+2q$, success $1-\beta$
& Polynomially superlinear proof bound and an explicit residual-entropy bound for near-linear proofs.\\[0.4em]
At most $r$ random bits and $q$ adaptive proof queries per claim
& Nonadaptive simulation with perfect completeness, $Q_{\mathrm{exp}}=1+2^{r+1}T_q(A)$
& A near-linear proof requires $Q_{\mathrm{exp}}=\Omega(\log K/\log\log K)$.\\[0.4em]
Binary alphabet, one proof query per claim
& Weak adaptive two-query decoder, success $1-\beta>1/2$
& $P=2^{\Omega(K)}$.\\
\bottomrule
\end{tabularx}
\caption{Three consequences of RLDC extraction from a common-test fiber.}
\label{tab:regimes}
\end{table}

The bounds leave several regimes open. Beyond the explicit deficit in Theorem~\ref{thm:quantitative-info}, they give no further restriction when $I(M;S)=K-o(K)$. They rely on a fixed alphabet, a fixed proof coordinate set, a positive robustness radius, and pointwise wrong-claim soundness. Any target-dependent state used to choose the tests must be represented by $S$ or by a variable $T$ that determines $S$. Ordinary dPCP list soundness alone does not imply the target-relative robustness assumed here.

\subsection{Open directions}

\begin{enumerate}[leftmargin=2.1em]
\item \textbf{Average-case robustness.} Can an information-to-coding reduction survive when wrong-claim soundness holds only for most coordinates or on average over the target distribution? Such a result would connect the present worst-coordinate interface more directly to standard list-sound formulations.

\item \textbf{Intermediate-information constructions.} The full-information endpoint is attainable, while a fixed linear information deficit forces a superlinear proof. Matching upper bounds or sharper converses in the regime $I(M;S)=K-r(K)$ would determine whether the explicit deficit in Theorem~\ref{thm:quantitative-info} is close to optimal.

\item \textbf{Adaptive lower-bound inputs.} Stronger lower bounds for weak adaptive RLDCs at more than two queries could avoid exhaustive decision-tree exposure and separate intrinsic adaptive query complexity from the cost of exact nonadaptive simulation.
\end{enumerate}

\section*{Declarations}

\paragraph{Funding.}
This work was supported by the 2025 Google Research Grant from Google Asia Pacific Pte. Ltd.

\paragraph{Competing interests.}
The author has no relevant financial or non-financial interests to disclose.

\paragraph{Author contributions.}
Hongmin Li conceived the study, developed and verified the mathematical arguments, conducted the literature analysis, and wrote the manuscript.

\paragraph{Acknowledgements.}
This research used the \href{https://utelecon.adm.u-tokyo.ac.jp/en/research_computing/utokyo_azure/}{UTokyo Azure service}.

\paragraph{Data and code availability.}
No datasets were generated or analyzed in this study. The results are mathematical; all arguments needed to verify them are contained in the article, and no custom code is required.

\paragraph{Use of artificial intelligence.}
Generative AI was used for language editing, structural feedback, literature searches, and checks for gaps in the proofs. The author reviewed all resulting suggestions and takes full responsibility for the manuscript.

\begingroup
\small
\setlength{\bibsep}{2pt}
\bibliographystyle{plainnat}
\bibliography{references}

@article{BGHSV2006,
  author  = {Ben-Sasson, Eli and Goldreich, Oded and Harsha, Prahladh and Sudan, Madhu and Vadhan, Salil P.},
  title   = {Robust {PCP}s of Proximity, Shorter {PCP}s, and Applications to Coding},
  journal = {SIAM Journal on Computing},
  volume  = {36},
  number  = {4},
  pages   = {889--974},
  year    = {2006},
  doi     = {10.1137/S0097539705446810}
}

@article{DallAgnolGurLachish2023,
  author  = {Dall'Agnol, Marcel and Gur, Tom and Lachish, Oded},
  title   = {A Structural Theorem for Local Algorithms with Applications to Coding, Testing, and Verification},
  journal = {SIAM Journal on Computing},
  volume  = {52},
  number  = {6},
  pages   = {1413--1463},
  year    = {2023},
  doi     = {10.1137/21M1422781}
}

@inproceedings{KatzTrevisan2000,
  author    = {Katz, Jonathan and Trevisan, Luca},
  title     = {On the Efficiency of Local Decoding Procedures for Error-Correcting Codes},
  booktitle = {Proceedings of the Thirty-Second Annual ACM Symposium on Theory of Computing},
  pages     = {80--86},
  year      = {2000},
  doi       = {10.1145/335305.335315}
}

@article{KerenidisDeWolf2004,
  author  = {Kerenidis, Iordanis and de Wolf, Ronald},
  title   = {Exponential Lower Bound for 2-Query Locally Decodable Codes via a Quantum Argument},
  journal = {Journal of Computer and System Sciences},
  volume  = {69},
  number  = {3},
  pages   = {395--420},
  year    = {2004},
  doi     = {10.1016/j.jcss.2004.04.007}
}

@article{GurLachish2021,
  author  = {Gur, Tom and Lachish, Oded},
  title   = {On the Power of Relaxed Local Decoding Algorithms},
  journal = {SIAM Journal on Computing},
  volume  = {50},
  number  = {2},
  pages   = {788--813},
  year    = {2021},
  doi     = {10.1137/19M1307834}
}

@article{Sauer1972,
  author  = {Sauer, Norbert},
  title   = {On the Density of Families of Sets},
  journal = {Journal of Combinatorial Theory, Series A},
  volume  = {13},
  number  = {1},
  pages   = {145--147},
  year    = {1972},
  doi     = {10.1016/0097-3165(72)90019-2}
}

@article{DinurMeir2011,
  author  = {Dinur, Irit and Meir, Or},
  title   = {Derandomized Parallel Repetition via Structured {PCP}s},
  journal = {Computational Complexity},
  volume  = {20},
  number  = {2},
  pages   = {207--327},
  year    = {2011},
  doi     = {10.1007/s00037-011-0013-5}
}

@article{DinurHarsha2013,
  author  = {Dinur, Irit and Harsha, Prahladh},
  title   = {Composition of Low-Error 2-Query {PCP}s Using Decodable {PCP}s},
  journal = {SIAM Journal on Computing},
  volume  = {42},
  number  = {6},
  pages   = {2452--2486},
  year    = {2013},
  doi     = {10.1137/100788161}
}

@inproceedings{ApplebaumNir2023,
  author    = {Applebaum, Benny and Nir, Oded},
  title     = {Advisor-Verifier-Prover Games and the Hardness of Information-Theoretic Cryptography},
  booktitle = {2023 IEEE 64th Annual Symposium on Foundations of Computer Science (FOCS)},
  pages     = {539--555},
  year      = {2023},
  doi       = {10.1109/FOCS57990.2023.00039}
}

@article{BoyleKomargodskiVafa2025,
  author  = {Boyle, Elette and Komargodski, Ilan and Vafa, Neekon},
  title   = {Memory Checking Requires Logarithmic Overhead},
  journal = {Journal of the ACM},
  volume  = {72},
  number  = {2},
  pages   = {13:1--13:43},
  year    = {2025},
  doi     = {10.1145/3707202}
}

@inproceedings{GoldbergGurSaraogi2026,
  author    = {Goldberg, Guy and Gur, Tom and Saraogi, Sidhant},
  title     = {Nearly Tight Lower Bounds for Relaxed Locally Decodable Codes via Robust Daisies},
  booktitle = {Proceedings of the 58th Annual ACM Symposium on Theory of Computing},
  pages     = {2210--2217},
  year      = {2026},
  doi       = {10.1145/3798129.3800923},
  eprint    = {2511.21659},
  archivePrefix = {arXiv}
}

@article{BlockEtAl2026,
  author  = {Block, Alexander R. and Blocki, Jeremiah and Cheng, Kuan and Grigorescu, Elena and Li, Xin and Zheng, Yu and Zhu, Minshen},
  title   = {Exponential Lower Bounds for 2-Query Relaxed Locally Decodable Codes},
  journal = {arXiv preprint arXiv:2602.20278},
  year    = {2026},
  eprint  = {2602.20278},
  archivePrefix = {arXiv},
  primaryClass = {cs.IT}
}

@article{ChengLiMao2026,
  author  = {Cheng, Kuan and Li, Xin and Mao, Songtao},
  title   = {When Relaxation Does Not Help: {RLDC}s with Small Soundness Yield {LDC}s},
  journal = {arXiv preprint arXiv:2603.03717},
  year    = {2026},
  eprint  = {2603.03717},
  archivePrefix = {arXiv},
  primaryClass = {cs.IT}
}

@inproceedings{GurMinzerWeissenbergZheng2026,
  author    = {Gur, Tom and Minzer, Dor and Weissenberg, Guy and Zheng, Kai Zhe},
  title     = {3-Query {RLDC}s Are Strictly Stronger Than 3-Query {LDC}s},
  booktitle = {Proceedings of the 58th Annual ACM Symposium on Theory of Computing},
  pages     = {1489--1496},
  year      = {2026},
  doi       = {10.1145/3798129.3800857},
  eprint    = {2512.12960},
  archivePrefix = {arXiv}
}
\endgroup

\end{document}